\documentclass[runningheads]{llncs}
\usepackage[T1]{fontenc}
\usepackage{graphicx}

\usepackage{amsmath}
\usepackage{amssymb}
\usepackage{mathtools}

\usepackage{algorithm} %
\usepackage{algpseudocode}
\floatname{algorithm}{Algorithm} %
\usepackage{float}
\usepackage{tikz} %
\usetikzlibrary{positioning}
\usepackage{tikz}
\usetikzlibrary{positioning,calc,arrows.meta}
\usepackage{enumitem} %
\usepackage{cite}
\setlist[itemize]{itemsep=3pt, parsep=0mm, leftmargin=1.5em}
\setlist[enumerate]{itemsep=3pt, parsep=0mm, leftmargin=1.5em}
\usepackage{amssymb}

\begin{document}
\title{Envy-Free School Redistricting\\Between Two Groups}
\author{Daisuke Shibatani \and Yutaro Yamaguchi}
\authorrunning{Daisuke Shibatani and Yutaro Yamaguchi}
\institute{Osaka University, Japan.\\\email{yutaro.yamaguchi@ist.osaka-u.ac.jp}}
\maketitle              %
\begin{abstract}

We study an application of fair division theory to school redistricting.
Procaccia, Robinson, and Tucker-Foltz (SODA 2024) recently proposed a mathematical model to generate redistricting plans that provide theoretically guaranteed fairness among demographic groups of students.
They showed that an almost proportional allocation can be found by adding $O(g \log g)$ extra seats in total, where $g$ is the number of groups.
In contrast, for three or more groups, adding $o(n)$ extra seats is not sufficient to obtain an almost envy-free allocation in general, where $n$ is the total number of students.
In this paper, we focus on the case of two groups.
We introduce a relevant relaxation of envy-freeness, termed \emph{1-relaxed envy-freeness}, which limits the capacity violation not in total but at each school to at most one.
We show that there always exists a 1-relaxed envy-free allocation, which can be found in polynomial time.
\end{abstract}

\section{Introduction}

\subsection{Background}

The public school a student attends is typically determined by their place of residence.
Even within the same municipality, significant disparities in school quality may exist, which makes the assignment of school attendance zones a crucial issue for students.
Due to several factors, municipalities must periodically restructure these zones.
However, school redistricting is often controversial and can sometimes cause severe conflicts. 

To mitigate such turmoil, Procaccia, Robinson, and Tucker-Foltz \cite{school} recently proposed a mathematical model to generate redistricting plans that provide theoretically guaranteed fairness among demographic groups based on fair division theory.
As noted in \cite{fair, survey}, problems related to fair division are ubiquitous, encompassing scenarios such as cake cutting, inheritance division, and chore allocation.
The problem of fair division has a very long history, %
and there are two primary concepts of fairness: proportionality and envy-freeness.
An allocation is considered \emph{proportional} if, given $n$ agents, every agent perceives their share to be worth at least $1/n$ of the total value.
An allocation is \emph{envy-free} if no agent strictly prefers another agent's share to their own.
While early research on fair division primarily focused on infinitely divisible goods, many recent studies have addressed fair division problems involving indivisible goods.

Procaccia et al.~\cite{school} applied fair division theory to school redistricting for the first time, defining relaxations of proportionality and envy-freeness specifically for this problem. 
They showed that an almost proportional allocation can always be found by adding $O(g \log g)$ seats in total, where $g$ is the number of groups; furthermore, when $g$ is constant, such an allocation can be found in polynomial time.
In contrast, they claimed that envy-freeness is difficult to achieve by demonstrating for $g \ge 3$ that adding $o(n)$ seats is not sufficient to obtain an almost envy-free allocation in general, where $n$ is the number of students.

\subsection{Our contribution}
In this study, we focus on the case where there exist two groups (when $g = 2$) and aim to establish a method to obtain an envy-free allocation.
Even in this case, as school seats are indivisible, we need to consider relaxation of envy-freeness. 
Specifically, we first introduce another reasonable relaxation of envy-freeness for the school redistricting problem, termed \emph{1-relaxed envy-freeness} (Definition~\ref{def:1-relax}), which only allows the capacity violation in \emph{each} school by at most one while the \emph{total} violation is bounded by a parameter $t$ in \cite{school}.
Although a 1-relaxed envy-free allocation does not necessarily exist when $g \ge 3$ as in the previous case \cite{school}, we show that when $g = 2$, it always exists and can be computed in polynomial time (Theorem~\ref{thm:a}).
The proof is based on the integrality of a network flow problem and ad-hoc transformation of solutions that are not 1-relaxed envy-free.

\subsection{Related work}
As reviewed by Suksompong \cite{Suksompong}, fair division has been extensively studied under various constraints.
However, none of these models are directly applicable to the setting of \cite{school}.
In fair division with several constraints, the allocation problem becomes significantly more complex --- or even impossible --- when the number of agents is three or more. 
For instance, in the continuous resource model (i.e., cake cutting), Stromquist \cite{Stromquist} proved that calculating a connected envy-free allocation for three or more agents using a finite number of queries is impossible.
For indivisible goods, under cardinality constraints with monotonic utility functions, Kyropoulou, Suksompong, and Voudouris \cite{Kyropoulou} proved the existence of an \emph{EF1 (envy-free up to one item)} allocation for two agents, whereas the case of three or more agents remains an open problem.
Suksompong \cite{Suksompong} outlines the differences between the cases of two agents and three or more agents under various other settings, including indivisible goods on a graph, graphical cake cutting, geometric constraints, and conflict constraints.
This illustrates that restricting the problem to two agents can sometimes enable the discovery of almost envy-free allocations that are intractable in the general case.

As mentioned above, \cite{school} is pioneering work on fair districting in the context of school redistricting.
Several other studies have addressed school redistricting, primarily focusing on optimization approaches.
For example, Allman, Ashlagi, Lo, Love, Mentzer, Ruiz-Sets, and O'Connell \cite{Allman} proposed an optimization framework aimed at preventing racial re-segregation and achieving diverse student demographics.

Closely related theoretical frameworks can also be found in \cite{Benabbou1, Benabbou2}, which consider matchings in weighted bipartite graphs with agents on one side and goods on the other.
Compared to these models, the setting in \cite{school} (as well as ours) is simpler in that all agents share identical valuations for the schools.
However, it introduces significant challenges because every student must be assigned to one of their accessible schools defined by geographical constraints. %

Subsequent to \cite{school}, Cembrano, Moraga, and Verdugo \cite{Cembrano} also investigated envy-freeness in the context of school redistricting. However, their notion of envy-freeness simply compares utilities under a specific allocation, ignoring geographical accessibility and seat-subset conditions. Therefore, their definition is fundamentally different from the concept of envy-freeness employed in \cite{school} and in our present study.

\section{Model and Result}
An instance of the \emph{school redistricting problem} is described by the following components:
\begin{itemize}
    \item A set of \emph{students} $ N = \{ 1,2,\dots,n \}$;
    \item A partition of the students into $g$ \emph{groups} $ N = N_1 \cup N_2 \cup \dots \cup N_g $;
    \item A set of \emph{schools} $M= \{ 1,2,\dots,m \}$, where each school $k$ has a \emph{capacity} $c_k$ and a (common) \emph{value} $v_k$;
    \item A bipartite graph $\mathcal{G} = (N \cup M, \mathcal{E})$ between students and schools, where the edge set $\mathcal{E}$ represents which students can be assigned to (i.e., can commute to) which schools.
\end{itemize}

An \emph{allocation} is a function $A\colon N \to M$ such that $\{ j, A(j) \} \in \mathcal{E}$ for every student $j \in N$.
Let $A^{-1}(k)=\{ j \in N \mid A(j) = k \}$ and $A^{-1}_i(k)= A^{-1}(k) \cap N_i$.
The \emph{utility} of group $i$ under an allocation $A$ is defined as $u(i,A) \coloneqq \sum_{j \in N_i} v_{A(j)}$.

The formulation up to this point is identical to that in \cite{school}.
In our model, we assume an \emph{initial} allocation $B$ is given, and the capacity of each school is set to $c_k = |B^{-1}(k)|$.
This $B$ can be interpreted as the current allocation or a tentative allocation only caring the number of students assigned to each school.
In either case, our goal is to correct unfairness of $B$, where it is desirable to assign almost the same number of students to each school as $B$.
In particular, we call an allocation $A$ \emph{amount-preserving} if $|A^{-1}(k)| = c_k$ for all $k \in M$.

Even when there are two groups, the indivisibility of school seats requires a relaxation of envy-freeness.
We define 1-relaxed envy-freeness as follows.

\begin{definition}\label{def:1-relax}
An allocation $A$ is \emph{$1$-relaxed envy-free} if it satisfies
\begin{align}\label{a}
\max_{k\in M} \left||A^{-1}(k)|-c_k\right| \le 1
\end{align}
and, for any pair of groups $(i_1,i_2)$, no other allocation $A'$ simultaneously satisfies the following three inequalities:
\begin{align}\label{b}
\max_{k\in M} \left||{A'}^{-1}(k)|-c_k\right|  &\le n\cdot \max_{k\in M} \left||A^{-1}(k)|-c_k\right|, \\ \label{c}
u(i_1,A') &> u(i_1,A), \\ \label{d}
|{A'}^{-1}_{i_1}(k)| &\le |A^{-1}_{i_2}(k)|  \quad (\forall k \in M). 
\end{align}
\end{definition}

First, \eqref{a} states that a capacity deviation of up to $1$ is allowed in each school under allocation $A$.
An allocation $A'$ satisfying \eqref{b}--\eqref{d} is considered to cause envy of group $i_1$ to $i_2$ in the allocation $A$, and we say $i_1$ \emph{envies} $i_2$ if there exists such an $A'$:
\begin{itemize}
    \item \eqref{b} ensures that a capacity deviation in $A'$ is permitted only to the extent that a deviation already exists in $A$.
    This means that, if $A$ is amount-preserving, then envy can be justified only by an amount-preserving allocation $A'$, but otherwise (if $A$ has a capacity deviation of at most one per school), it can be justified by any allocation $A'$ (which can have an arbitrarily large deviation).
    \item \eqref{c} indicates that the utility of group $i_1$ strictly increases.
    \item \eqref{d} means that the seats assigned to $i_1$ under $A'$ can be regarded as a subset of the seats assigned to $i_2$ under $A$.
\end{itemize}

The notion of 1-relaxed envy-freeness differs from $t$-envy-freeness in \cite{school} in two senses.
One is how capacity violations are measured: our definition imposes per-school bounds (evaluated both above and below), whereas only the total violation is bounded (only from above) by a parameter $t$ in \cite{school}.
The other is how the utilities of $i_1$ under $A$ and $A'$ are compared: our definition simply compares them as \eqref{c}, whereas the inequality is also relaxed in \cite{school} by adding a $\Theta(t)$ term in the right-hand side (which requires a large gap of utilities to justify envy of $i_1$).
Thus, if the capacity of each school is defined in the same way, a 1-relaxed envy-free allocation in our definition can be translated into an $(m/2)$-envy-free allocation in the sense of \cite{school}.
Also, using essentially the same argument as in the proof of \cite[Theorem 3.1]{school}, we can observe that a 1-relaxed envy-free allocation does not necessarily exist when $g \ge 3$.

The main result of this paper is stated as follows:

\begin{theorem}\label{thm:a}
For any instance of the school redistricting problem with two groups, a 1-relaxed envy-free allocation always exists.
Furthermore, such an allocation can be constructed in polynomial time.
\end{theorem}

\section{Proofs}
In this section, we prove Theorem~\ref{thm:a}.
We prove it separately for the cases where the sizes of two groups are different (in Section~\ref{sec:easy}) and the same (in Section~\ref{sec:main}).

\subsection{Easy case: When two groups have different sizes}\label{sec:easy}
When $|N_1| \neq |N_2|$, the following lemma completes the proof.

\begin{lemma}\label{lem:a}
If the two groups have different sizes, an amount-preserving allocation maximizing the utility of the smaller group is always 1-relaxed envy-free.
Furthermore, such an allocation can be computed in polynomial time.
\end{lemma}

\begin{proof}
In this setting, since an initial allocation is given and gives the capacity of each school, every instance is feasible.
Let group 1 be the smaller group and group 2 be the larger group among the two groups.
Let $A^*$ be a utility-maximizing amount-preserving allocation for group 1; that is, $A^*$ is an allocation that maximizes $u(1, A^*)$ subject to $|{A^*}^{-1}(k)| = c_k$ for all $k \in M$.

First, we show that group 2 (larger) does not envy group 1 (smaller).
From \eqref{d} in Definition~\ref{def:1-relax}, if group 2 envies group 1, there exists an allocation $A'$ satisfying
\begin{align*}
|{{A'}^{-1}_{2}}(k)| \le |{{A^*}^{-1}_{1}}(k)|\quad (\forall k \in M).
\end{align*}
By summing this up, we obtain
\begin{align*}
|N_2| = \sum_{k\in M}|{{A'}^{-1}_{2}}(k)| \le \sum_{k\in M}|{A^*}^{-1}_{1}(k)|=|N_1|,
\end{align*}
but this contradicts the fact that the size of group 2 is strictly greater than the size of group 1. %
Thus, group 2 cannot envy group 1.

Next, we show that group 1 does not envy group 2.
From \eqref{b} and \eqref{c} in Definition~\ref{def:1-relax}, if group 1 envies group 2, there exists an amount-preserving allocation $A'$ satisfying
\begin{align*}
u(1,A') > u(1,A^*),
\end{align*}
but this contradicts that $A^*$ maximizes the utility of group 1 under the capacity constraint.
Thus, group 1 cannot envy group 2.

From the above, we have shown that the allocation $A^*$ is 1-relaxed envy-free.
Furthermore, as in \cite{school}, such an ideal allocation $A^*$ for group 1 can be computed in polynomial time using the Hungarian method \cite{Kuhn,Munkres} (by reducing to the so-called weighted $b$-matching problem; see, e.g., \cite{Combinatorial}). \qed
\end{proof}

\subsection{Main case: When two groups have the same size}\label{sec:main}
In what follows, we assume $|N_1| = |N_2|$.
We first show a lemma that guarantees the existence of a symmetric allocation for any not 1-relaxed envy-free amount-preserving allocation.
Here, for a given allocation $A$, an allocation $A'$ is called a \emph{perfectly-swapped} allocation with respect to $A$ if it satisfies $|{{A'}^{-1}_{1}}(k)| = |{A^{-1}_{2}}(k)|$ and $|{{A'}^{-1}_{2}}(k)| = |{A^{-1}_{1}}(k)|$ for all $k \in M$.

\begin{lemma}\label{lem:b}
Let $A$ be an amount-preserving allocation.
If there is no perfectly-swapped allocation $A'$ with respect to $A$, then $A$ is 1-relaxed envy-free.
\end{lemma}

\begin{proof}
We prove the contrapositive; that is, if the allocation $A$ is not 1-relaxed envy-free, then there exists a perfectly-swapped allocation $A'$.
Suppose that $A$ is not 1-relaxed envy-free.
By symmetry, we assume without loss of generality that group 1 envies group 2.
Let this envy in allocation $A$ be caused by the existence of some allocation $A'$ that simultaneously satisfies \eqref{b}--\eqref{d} in Definition~\ref{def:1-relax}.
Since $A$ is amount-preserving, this should also be true for $A'$ by \eqref{b}.
Furthermore, as the sizes of the two groups are equal, the following holds:
\begin{align*}
\sum_{k\in M}|{{A'}^{-1}_{1}}(k)| = |N_1| = |N_2| = \sum_{k\in M}|{A^{-1}_{2}}(k)|.
\end{align*}
This implies that \eqref{d} holds with equality for all $k \in M$, i.e.,
\begin{align}\label{h}
|{{A'}^{-1}_{1}}(k)| = |{A^{-1}_{2}}(k)| \quad (\forall k \in M).
\end{align}
Also, the total number of students assigned to each school is the same in both allocations $A$ and $A'$ (as they are both amount-preserving), yielding the following:
\begin{align}\label{i}
|{A^{-1}_{1}}(k)| + |{A^{-1}_{2}}(k)| = c_k = |{{A'}^{-1}_{1}}(k)|+ |{{A'}^{-1}_{2}}(k)| \quad (\forall k \in M).
\end{align}
From \eqref{h} and \eqref{i}, we have the following:
\begin{align*}
|{{A'}^{-1}_{2}}(k)| = |{A^{-1}_{1}}(k)| \quad (\forall k\in M).
\end{align*}
Thus, $A'$ is a perfectly-swapped allocation with respect to $A$. \qed
\end{proof}

Let $B$ be the initial allocation.
By Lemma~\ref{lem:b}, if there is no perfectly-swapped allocation with respect to $B$, then $B$ is 1-relaxed envy-free.
Finding such an allocation is simply reduced to the maximum flow problem (more specifically, the so-called $b$-matching problem; see, e.g., \cite{Combinatorial}) and solved in polynomial time \cite{FF1956}.
Thus, we may assume that there exists a perfectly-swapped allocation $B'$.
The following lemma enables us to make a stronger assumption: there exists an amount-preserving allocation $A$ and a perfectly-swapped allocation $A'$ such that
\[|A_1^{-1}(k)| - |A_2^{-1}(k)| = |{A'}_2^{-1}(k)| - |{A'}_1^{-1}(k)| \in \{-1, 0, 1\} \quad (\forall k \in M).\]

\begin{lemma}\label{lem:d}
Given a perfectly-swapped allocation $B'$ with respect to the initial allocation $B$, one can compute an amount-preserving allocation $A$ such that $\left||A_1^{-1}(k)| - |A_2^{-1}(k)|\right| \le 1$ for all $k \in M$.
\end{lemma}

\begin{proof}
Let $N = \{ i_1,i_2,\dots,i_n \}$ be the set of students, and let $M = \{s_1,s_2,\dots,s_m \}$ be the set of schools.
Let $g(j) \in \{1, 2\}$ denote the group of each student $j \in N$.
We construct an instance of the network flow problem as follows (see Fig.~\ref{fig:a}):
\begin{itemize}
 \item Vertex set $V = \{r^+\} \cup N \cup (M \times \{1, 2\}) \cup M \cup \{r^-\}$:\\
 $r^+$: source, $j \in N$: student vertices, $(k,1), (k,2) \in M \times \{1, 2\}$: school-group pair vertices, $k\in M$: school vertices, $r^-$: sink.
 \item Edge set $E = E_s \cup E_N \cup E_G \cup E_M \cup E_r$ and capacity (flow lower bound $\ell$ and upper bound $u$):
 \begin{itemize}
  \item From source $r^+$ to each student:
  \begin{itemize}
      \item $E_s=\{ (r^+, j)\mid j\in N\}$;
      \item $\ell(e)= u(e)=1$.
  \end{itemize}
 \item From each student to school-group pairs:
 \begin{itemize}
  \item $E_N=\{ (j,(k,g(j)))\mid j\in N,\ k\in \{ B(j), B'(j)\}\}$ (if $B(j)=B'(j)$, we allow multiple edges);
  \item $\ell(e)=0$, $u(e)=1$.
 \end{itemize}
  \item From school-group pairs to each school:
  \begin{itemize}
   \item $E_G=\{((k,h),k)\mid k\in M,\ h\in \{1,2\}\}$;
   \item $\ell(e) = \lfloor c_k/2 \rfloor$, $u(e) = \lceil c_k/2 \rceil$.
  \end{itemize}
   \item From each school to sink $r^-$:
  \begin{itemize}
   \item $E_M=\{ (k, r^-)\mid k\in M\}$;
   \item $\ell(e)= u(e)=c_k$.
  \end{itemize}
   \item From sink $r^-$ to source $r^+$:
   \begin{itemize}
    \item $E_r=\{(r^-, r^+)\}$;
    \item $\ell(e) = u (e)=n$.
   \end{itemize}
  \end{itemize}
 \end{itemize}

\begin{figure}[t!]
\centering

\begin{tikzpicture}[
    scale=0.8, transform shape, %
    >={Stealth[length=3mm]}, %
    node distance=1.2cm and 2.0cm, %
    every node/.style={font=\large}, %
    label_text/.style={font=\small, above, sloped} %
]

    \node (s0) at (0,0) {$r^+$};

    \node (i2) [right=of s0, yshift=0.5cm] {$i_2$};
    \node (i1) [above=1.2cm of i2] {$i_1$};
    \node (dots_i) [below=0.5cm of i2] {$\vdots$};
    \node (in) [below=0.8cm of dots_i] {$i_n$};

    \node (s1_2) [right=2.5cm of i1, yshift=-0.4cm] {$(s_1, 2)$};
    \node (s1_1) [above=0.1cm of s1_2] {$(s_1, 1)$};
    
    \node (s2_1) [right=2.5cm of i2, yshift=0.4cm] {$(s_2, 1)$};
    \node (s2_2) [below=0.1cm of s2_1] {$(s_2, 2)$};

    \node (dots_pair) at (dots_i -| s1_1) {$\vdots$}; %

    \node (sm_1) [right=2.5cm of in, yshift=0.4cm] {$(s_m, 1)$};
    \node (sm_2) [below=0.1cm of sm_1] {$(s_m, 2)$};

    \node (s1) [right=2.5cm of s1_2, yshift=0.3cm] {$s_1$};
    \node (s2) [right=2.5cm of s2_1, yshift=-0.3cm] {$s_2$};
    \node (dots_s) at (dots_pair -| s1) {$\vdots$};
    \node (sm) [right=2.5cm of sm_1, yshift=-0.3cm] {$s_m$};

    \node (t0) at ($(s0.east) + (12cm,0)$ |- s0) {$r^-$}; %

    \draw[->] (s0) -- node[label_text, midway] {$[1,1]$} (i1);
    \draw[->] (s0) -- (i2);
    \draw[->] (s0) -- (in);

    \draw[->] (i1) -- node[label_text, midway] {$[0,1]$} (s1_1);
    \draw[->] (i1) -- (s2_1); %
    
    \draw[->] (i2) -- (s1_2); %
    \draw[->] (i2) -- (dots_pair);

    \draw[->] (in) -- (dots_pair);
    \draw[->] (in) -- (sm_2);

    \draw[->] (s1_1) -- (s1);
    \draw[->] (s1_2) -- (s1);
    \node[font=\small] at ($(s1_1)!0.5!(s1) + (0.4, 0.5)$) {$[\lfloor c_k/2 \rfloor, \lceil c_k/2 \rceil]$};

    \draw[->] (s2_1) -- (s2);
    \draw[->] (s2_2) -- (s2);

    \draw[->] (sm_1) -- (sm);
    \draw[->] (sm_2) -- (sm);

    \draw[->] (s1) -- node[label_text] {$[c_k, c_k]$} (t0);
    \draw[->] (s2) -- (t0);
    \draw[->] (sm) -- (t0);

    \draw[->, line width=0.6pt] (t0) to[out=100, in=80, looseness=1.1] node[above] {$[n, n]$} (s0);

\end{tikzpicture}
\caption{Conceptual diagram of the network in Lemma \ref{lem:d}}
\label{fig:a}
\end{figure}
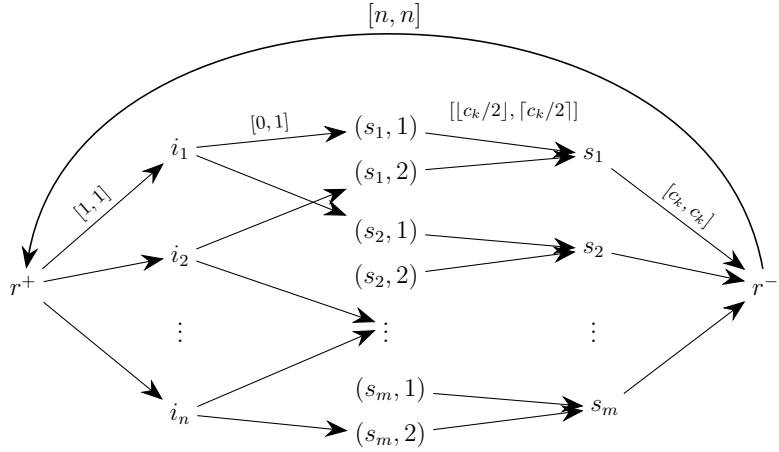
 
In the following, we consider a \emph{circulation} in this network, which is a function $f \colon E \to \mathbb{R}$ satisfying the following:
\begin{itemize}
    \item[] \textbf{Capacity constraint:} $\ell(e) \le f(e) \le u(e)$ for all edges $e \in E$;
    \item[] \textbf{Flow conservation:} $\sum_{e \in \delta^+(v)} f(e) = \sum_{e \in \delta^-(v)} f(e)$ for all vertices $v \in V$, where $\delta^+(v)$ and $\delta^-(v)$ denote the sets of edges outgoing from and incoming to $v$, respectively.
\end{itemize}
As $B$ and $B'$ are both amount-preserving allocations, we can construct a circulation $f$ as follows:
\begin{itemize}
 \item $f(e) = 1$ for each $e \in E_s$;
 \item $f(e) = 1/2$ for each $e \in E_N$;
 \item $f(e) = c_k/2$ for each $e \in E_G$;
 \item $f(e) = c_k$ for each $e \in E_M$;
 \item $f(e) = n$ for each $e \in E_r$.
\end{itemize}
The capacity constraint is obviously satisfied, and the flow conservation is confirmed by regarding the assignment of student $j$ to school $k$ under allocation $A \in \{B, B'\}$ as a circulation of value $1/2$ along a directed cycle $r^+ \to j \to (k, g(j)) \to k \to r^- \to r^+$ (note that $f$ is the sum of these circulations).

Consequently, by the integrality of the network flow problem (cf.~\cite[Theorem 11.2]{Combinatorial}), we conclude that an integer circulation $\tilde{f} \colon E \to \mathbb{Z}$ exists in the same network.
Furthermore, such an integer circulation can be found by a single maximum-flow computation in an auxiliary network (cf.~\cite[Theorem 11.3]{Combinatorial}); thus, it can be computed in polynomial time \cite{FF1956}.
By the definition of the network ($E_s$ and $E_r$ in particular), the obtained integer circulation $\tilde{f}$ is decomposed into $n$ circulations, each of value $1$ along a directed cycle of form $r^+ \to j \to (k, g(j)) \to k \to r^- \to r^+$ as above.
Then, an allocation $A$ is naturally determined so that each student $j$ is assigned to school $k$ along the unique relevant cycle.
By the definition of capacity on $E_G$ and $E_M$, under $A$, each school has exactly $\lfloor c_k/2 \rfloor$ students of one group and exactly $\lceil c_k / 2 \rceil$ students of the other group, whose difference is at most $1$.
This completes the proof. \qed
\end{proof}

Let $A$ be the allocation obtained by Lemma \ref{lem:d}.
If a perfectly-swapped allocation $A'$ does not exist, then $A$ is already a 1-relaxed envy-free allocation by Lemma \ref{lem:b}.
Thus, we assume that such an $A'$ exists.
Then, by appropriately modifying $A$ and $A'$, we can obtain a 1-relaxed envy-free allocation as follows, which completes the proof of Theorem~\ref{thm:a}:

\begin{lemma}\label{lem:e}
Given an amount-preserving allocation $A$ and a perfectly-swapped allocation $A'$ such that $\left||A_1^{-1}(k)| - |A_2^{-1}(k)|\right| = \left||{A'}_1^{-1}(k)| - |{A'}_2^{-1}(k)|\right| \le 1$ for all $k \in M$, one can compute an allocation $X$ such that $\left||X^{-1}(k)| - c_k\right| \le 1$ and $|X_1^{-1}(k)| = |X_2^{-1}(k)|$ for all $k \in M$, which is 1-relaxed envy-free.
\end{lemma}
\begin{proof}
Let us define $\Delta_{A, k} \coloneqq |A^{-1}_1(k)| - |A^{-1}_2(k)| \in \{-1, 0, 1\}$.
We call a school with $\Delta_{A, k} = -1$ a \emph{deficient} school, and a school with $\Delta_{A, k} = 1$ an \emph{excess} school.

We obtain a desired allocation $X$ by Algorithm~\ref{alg:adjustment}.
Starting with $X = A$, this algorithm resolves pairs of excess and deficient schools one by one by moving students of group 1.
To do this, considering only group 1, we construct a directed graph $G_X = (M, E_X)$ on the set of schools by
\begin{align}\label{eq:E_X}
  E_X \coloneqq \{ e_j = (X(j), A'(j)) \mid j \in N_1\},  
\end{align}
where $G_X$ may contain a self-loop when $X(j) = A'(j)$ for some $j \in N_1$.
We move students along a directed path in $G_X$.
A student $j \in N$ moved in each iteration is fixed to their destination school; that is, we set $X(j) \leftarrow A'(j)$, who will never be moved further.
Let $S_X$ and $T_X$ denote the sets of excess schools and deficient schools, respectively, with respect to the currect allocation $X$.

\begin{algorithm}[t]\caption{Adjustment to a balanced allocation}
\label{alg:adjustment}
\begin{algorithmic}
\Require Allocations $A$ and $A'$ satisfying the condition in Lemma~\ref{lem:e}.
\Ensure An allocation $X$ satisfying the condition in Lemma~\ref{lem:e}.
\State \textbf{Step 1 (Initialization):} 
\Statex \quad Set $X \leftarrow A$ and construct the directed graph $G_X$.

\State \textbf{Step 2 (Search and update):}
\Statex \quad Repeat the following operations until $S_X$ becomes empty:
\State \quad (a) Select an excess school $s \in S_X$.
\State \quad (b) Find a directed path $P = (s = v_0, v_1, \dots, v_\ell = t)$ in $G_X$ from $s$ to some deficient school $t \in T_X$.
\State \quad (c) For each edge $e_j = (v_i, v_{i+1})$ on $P$, update $X(j) \leftarrow A'(j)$ (originally, $X(j) = v_i$ and $A'(j) = v_{i+1}$).
After that, update $G_X$ (as well as $S_X$ and $T_X$) with respect to the new $X$.

\State \textbf{Step 3 (Termination):}
\Statex \quad Return $X$.
\end{algorithmic}
\end{algorithm}

The proof is completed as follows:
Lemma~\ref{lem:f} guarantees that there always exists a path $P$ from any $s \in S_X$ to $T_X$ in $G_X$ when $S_X \neq \emptyset$.
After that, we see that the algorithm terminates in polynomial time in Lemma~\ref{lem:g} and that its output $X$ is indeed a desired allocation in Lemma~\ref{lem:h}. \qed
\end{proof}

\begin{lemma}\label{lem:f}
Let $X$ be the allocation during an iteration of Algorithm~\ref{alg:adjustment}.
Fix any excess school $s \in S_X$, and define the set of schools reachable from $s$ in $G_X$ as
\begin{align*}
R \coloneqq \{k\in M \mid s \rightsquigarrow k\}.
\end{align*}
Then, $R$ always contains a deficient school in $T_X$.
\end{lemma}

\begin{proof}
Recall that $\Delta_{X, k} = |X^{-1}_1(k)| - |X^{-1}_2(k)|$, and by definition of $G_X$, the out-degree $d_X^+(k)$ and in-degree $d_X^-(k)$ of each school $k\in M$ can be written as follows:
\begin{align}\label{i2}
d_X^+(k) = |X_1^{-1}(k)|, \quad d_X^-(k) = |{A'}_1^{-1}(k)|.
\end{align}
Since $A'$ is a perfectly-swapped allocation with respect to $A$ and the students of group 2 are not moved, we have
\begin{align}\label{k}
|{A'}_1^{-1}(k)| = |{A}_2^{-1}(k)| = |X_2^{-1}(k)|.
\end{align}
From \eqref{i2} and \eqref{k}, we obtain the following:
\begin{align}\label{l}
d_X^+(k) - d_X^-(k) = |X_1^{-1}(k)| - |X_2^{-1}(k)| = \Delta_{X, k}.
\end{align}

Suppose for the sake of contradiction that $R$ does not contain a deficient school.
Then, we have
\begin{align}\label{m}
\Delta_{X, s} = 1, \quad \Delta_{X, k} \ge 0 \quad (\forall k\in R).
\end{align} 
From \eqref{l} and \eqref{m}, as $s \in R$, we obtain
\begin{align}
\sum_{k\in R}\bigl(d_X^+(k) - d_X^-(k)\bigr)= \sum_{k\in R}\Delta_{X, k} \ge 1,
\end{align}
which can be rephrased as
\begin{align}
\sum_{k\in R} d_X^+(k) > \sum_{k\in R} d_X^-(k).\label{eq:neg}
\end{align}
On the other hand, by definition of $R$, no edge in $G_X$ leaves from $R$ to outside.
Thus, counting the edges incident to $R$, we obtain
\begin{align*}
\sum_{k\in R} d_X^+(k) \le \sum_{k\in R} d_X^-(k),%
\end{align*}
which contradicts \eqref{eq:neg}.
Therefore, we have shown that $R$ contains at least one deficient school. \qed
\end{proof}

\begin{lemma}\label{lem:g}
Algorithm \ref{alg:adjustment} terminates in polynomial time.
\end{lemma}

\begin{proof}
From Lemma \ref{lem:f}, in each iteration of Step 2 of Algorithm \ref{alg:adjustment}, there exists a path from $s$ to some $t \in T_X$ in $G_X$, which is easily found in linear time.
By moving students of group 1 along the path $P$, the value $\Phi_X \coloneqq \sum_{k \in M}|\Delta_{X,k}|$ decreases by $2$ as follows.
For each inner vertices $v_i$ on $P$, exactly one student of group $1$ is moved out and exactly one is moved in, so $|X^{-1}_1(v_i)|$ does not change before and after the update in (c).
In contrast, for the end vertices $s$ and $t$, exactly one student of group $1$ is moved out and moved in, respectively, so $|X^{-1}_1(s)|$ decreases by one and $|X^{-1}_1(t)|$ increases by one, which makes $|\Delta_{X,s}|$ and $|\Delta_{X,t}|$ from $1$ to $0$.
Thus, $\Phi_X$ becomes $0$ in at most $\Phi_A / 2 \le m/2$ iterations.
Thus, Algorithm \ref{alg:adjustment} terminates in polynomial time. \qed
\end{proof}

\begin{lemma}\label{lem:h}
    The output $X$ of Algorithm \ref{alg:adjustment} satisfies $\left||X^{-1}(k)| - c_k\right| \le 1$ and $|X_1^{-1}(k)| = |X_2^{-1}(k)|$ for all $k \in M$, which is 1-relaxed envy-free.
\end{lemma}

\begin{proof}
Recall that at the beginning of Algorithm~\ref{alg:adjustment}, $X = A$ is amount-preserving and satisfies $\Delta_{X,k} \in \{-1, 0, 1\}$ for all $k \in M$.
As described in the proof of Lemma~\ref{lem:g}, in each iteration of Step 2, the number of students of group 1 decreases by one at the start vertex $s$ of the path $P$, increases by one at the goal vertex $t$, and remains unchanged at the inner vertices.
Since each school $k$ with $|\Delta_{X, k}| = 1$ is chosen as an end vertex of $P$ exactly once and any student of group 2 is not moved in the algorithm, we have $\left||X^{-1}(k)| - c_k\right| \le 1$ for all $k \in M$.
Also, at the end of Algorithm \ref{alg:adjustment}, the allocation $X$ satisfies $\Delta_{X, k} = 0$ for all $k \in M$, which means that $|X^{-1}_1(k)| = |X^{-1}_2(k)|$ for all $k \in M$.

Consequently, the utilities of the two groups become equal:
\begin{align}\label{u}
u(1,X) = u(2,X).
\end{align}
Suppose for the sake of contradiction that $X$ is not 1-relaxed envy-free.
This means that one group envies the other.
By symmetry, we can assume without loss of generality that group 1 envies group 2.
Let this envy in allocation $X$ be caused by the existence of some allocation $X'$ that simultaneously satisfies \eqref{b}--\eqref{d} in Definition~\ref{def:1-relax}; in particular, \eqref{c} and \eqref{d} are written as follows:
\begin{align}\label{x}
    u(1, X') &> u(1, X), \\ \label{y}
    |{{X'}^{-1}_{1}}(k)|&\le |{X^{-1}_{2}}(k)| \quad (\forall k \in M).
\end{align}
Since the two groups have the same size, we have the following:
\begin{align}\label{z}
    \sum_{k\in M}|{{X'}^{-1}_{1}}(k)| =\sum_{k\in M}|{X^{-1}_{1}}  (k)| = \sum_{k\in M}|{X^{-1}_{2}}(k)|.
\end{align}
From \eqref{y} and \eqref{z}, we obtain
\begin{align}\label{v}
 |{{X'}^{-1}_{1}}(k)|=|{X^{-1}_{2}}(k)| \quad (\forall k \in M).
\end{align}
Thus, by \eqref{u} and \eqref{v} (as the utility is just defined by the sum of the values of assigned schools), the following holds:
\begin{align}
u(1,X') = u(2,X) = u(1,X),
\end{align}
which contradicts \eqref{x}.
Therefore, $X$ is a 1-relaxed envy-free allocation. \qed
\end{proof}

\section{Conclusion and Future Work}
In this study, we have addressed the fair division problem in school redistricting. 
While Procaccia et al.~\cite{school} showed for three or more groups that adding $o(n)$ extra seats in total and relaxing how to compare the utilities are not sufficient to guarantee the existence of an envy-free allocation in general, we have proved for the case of two groups that reasonable envy-freeness is achieved with at most one seat deviation in each school. 
Specifically, we have introduced the 1-relaxed envy-freeness as a relaxation of envy-freeness tailored to the setting where we are given an initial allocation of the students to the schools.
We have shown that such a 1-relaxed envy-free allocation always exists and can be computed in polynomial time.

Although we have shown positive results for two groups, the setting may be considered restrictive in the following senses:
\begin{itemize}
    \item In both our model and the previous work \cite{school}, the value of each school is the same for all students. We can relax this assumption so that the value of each school is the same for all students in each group (which can be different over the groups), but it is still restrictive because the students' preferences may be individually different even in the same group.
    \item In both our model and the previous work \cite{school}, envy of group $i_1$ to group $i_2$ is justified only by another allocation where the seats for $i_1$ in every school is included by the seats for $i_2$ in the school under the current allocation (cf.~\eqref{d} in Definition~\ref{def:1-relax}).
    This makes envy too difficult to occur; in particular, a larger group $i_1$ cannot envy a smaller group $i_2$ (cf.~the proof of Lemma~\ref{lem:a}).
    \item In our problem setting, the capacity of each school is defined as the number of students in the initial allocation. The previous work \cite{school} considers redistricting scenarios driven by demographic changes or capacity fluctuations due to school opening or closing. Our present results can be applied after computing a tentative allocation $B$ that only cares the number of students assigned to each school, but it does not give any direct implication to their setting in general.
\end{itemize}
It is a possible direction of future work to formulate and investigate a wider reasonable setting of envy-free school redistricting to address these restrictive aspects.

\begin{credits}
\subsubsection{\ackname}
This work was supported by JSPS KAKENHI Grant Number JP25H01114 and JST CRONOS Japan Grant Number JPMJCS24K2.

\end{credits}

\end{document}